\documentclass[conference,letterpaper]{ieeeconf} 

\usepackage{amsmath,amssymb,amsthm}   
\usepackage{dsfont, graphicx,color}
\usepackage{xcolor}
\usepackage{mathrsfs, halloweenmath}

\IEEEoverridecommandlockouts                              

\newcommand{\abs}[1]{\left\vert{#1}\right\vert}

\newcommand{\vo}[1]{\overline{{#1}}}

\newcommand{\R}{\mathds R}

\title{
In-phase oscillations from the cooperation of cellular and network positive feedback in synaptically-coupled oscillators
}

\theoremstyle{definition}          
\newtheorem{lem}{Lemma}
\newtheorem{defn}{Definition}

\theoremstyle{definition}
\newtheorem{theo}{Theorem}

\author{Omar Juarez-Alvarez$^{1}$ and Alessio Franci$^{2}$
\thanks{This work was supported by UNAM-DGAPA-PAPIIT grant IN102420 and by CONACyT grant A1-S-10610.}
\thanks{$^{1}$O. P. Juarez-Alvarez is with the Department of Mathematics, Faculty of Sciences, UNAM.        {\tt\small pat\_ jualv@ciencias.unam.mx}}%
\thanks{$^{2}$A. Franci is with the Department of Mathematics, Faculty of Sciences, UNAM,
        {\tt\small afranci@ciencias.unam.mx}.
        }%
}

\begin{document}

\maketitle
\thispagestyle{empty}
\pagestyle{empty}


\begin{abstract}
We study the emergent dynamics of a network of synaptically coupled slow-fast oscillators. Synaptic coupling provides a network-level positive feedback mechanism that cooperates with cellular-level positive feedback to ignite in-phase network oscillations. Using analytical bifurcation analysis, we prove that the Perron-Frobenius eigenvector of the network adjacency matrix fully controls the oscillation pattern {locally in a neighborhood of a Hopf bifurcation}. Besides shifting the focus from the spectral properties of the network Laplacian matrix to the network adjacency matrix, we discuss other key differences between synaptic and diffusive coupling.

\end{abstract}

\section{Introduction}

Synchronization is usually studied in the context of diffusive coupling \cite{kim16, panteley17}, \emph{i.e.}, when the interaction between the oscillators is proportional to the difference in their states. Focusing on diffusive coupling has various limitations. First, because diffusive coupling is passive, oscillators must be intrinsic, that is, they must exhibit limit cycle oscillations in the absence of network interactions. Second, the only type of emergent network activity is a practically synchronous one, where, for large enough diffusive coupling strength, the oscillators converge to the same state modulo a synchronization error.  Motivated by understanding the emergence of sustained in-phase oscillations in the suprachiasmatic nucleus (SCN) in the mammal master circadian clock \cite{arechiga04, evans16}, we introduce a model of slow-fast oscillators with synaptic-like coupling and study the emergence of in-phase oscillations in it. A fundamental experimental observation, reproduced in our model but impossible to reproduce in diffusively coupĺed models, is that in the SCN many clock neurons behave as sustained oscillators only in the presence of network interactions, whereas they behave as damped oscillators when isolated \cite{webb09}.
SCN dynamics are therefore emergent, in the sense that the collective behavior (sustained oscillations) relies on network interactions and it is distinctively different from the isolated node behavior (damped oscillations) \cite{noguchi17}.\footnote{Note that the notion of emergent dynamics used in~\cite{kim16, panteley17} is different from ours.}

The intrinsic dynamics of our oscillators include a saturated fast cellular positive feedback loop and a linear slow negative feedback loop. It is a simplified version of excitable neural dynamics \cite{hh, fhn}. The interaction of the two loops leads to relaxation (slow-fast) neural-like oscillations for strong enough positive feedback through a Hopf bifurcation. The network synaptic-like couplings are approximated as saturated inputs to the receiving oscillator of the state of the sending oscillator. They provide network positive feedback, which cooperates with cellular positive feedback to ignite and shape emergent network oscillations.

The contributions of our analysis are the following. First, we prove a general lemma for the spectral properties of a class of block-defined matrices with the structure of the Jacobian matrix of our model. Second, we show that diffusive coupling cannot induce synchronous oscillations in a network of damped oscillators. Third, we prove that under a strongly directed network topology synaptic coupling can lead to in-phase oscillations even when the uncoupled oscillators are damped. If the coupling is in-regular, in-phase oscillations become synchronous, \emph{i.e.}, all oscillators converge to the same state. In this work, we rely on (local) bifurcation analysis at the model equilibrium and show that the (dominant) Perron-Frobenius eigenvector of the network adjacency matrix fully determines the in-phase oscillation pattern. Our results are in line with existing ones on automata synchronization \cite{gusev16, zhong17} and, together with \cite{lee20b}, they stress the importance of considering non-diffusive coupling in synchronization studies. In future works, we will couple our local results with a global analysis, using, for instance, dominance analysis~\cite{forni2018differential}. Also, we only consider here homogeneous (identical) intrinsic dynamics. In future works we will relax this assumption as well by exploiting the power of synaptic coupling of being naturally apt to cope with non-synchronous in-phase oscillations, as those that are expected in heterogeneous populations.

\section{Notation and definitions}

$\mathds{N}$ denotes the set of positive natural numbers, and $\mathds{R}$ the set of real numbers. In general, $N\in\mathds{N}$ will be a positive integer. As usual, $\mathrm{Re}(z)=x$ denotes the real part of a complex number $z=x+iy\in\mathds{C}$. $\mathds{R}^N$ denotes the set of real $N$-tuples, and $\vo{x}\in\R^N$ denotes an arbitrary $N$-tuple. Because of the specific models used, it will be convenient to denote $\mathds{R}^{2N}=\mathds{R}^N\times\mathds{R}^N$ and its elements as $(\vo{x},\vo{y})\in\mathds{R}^{2N}$. The zero  and one vectors $\vo{0}_N\in\mathds{R}^N$, $\vo{1}_N\in\mathds{N}$, denote tuples which have all their entries equal to zeroes and ones, respectively. Finally, a vector is said to be \emph{positive} if all its entries are strictly positive, denoted by $\vo{x}>0$.

A {\it sigmoid} is a bounded, continuously differentiable function $S:\mathds{R}\to\mathds{R}$ such that $S(0)=0$, $S'(x)>0$ for all $x\in\R$, $S'(0)=1$, and ${\rm argmax}_{x\in\R}S'(x)=0$.

\par The set $\mathscr{M}_{N\times N}$ contains all real $N\times N$ matrices represented as $M=(M_{ij})$. $I_N$ denotes the identity matrix in dimension $N$, and $O_N$ denotes the zero matrix in dimension $N$. The determinant of a matrix $A\in\mathscr{M}_{N\times N}$ is denoted by $\abs{A}$, and its characteristic polynomial is denoted by $p(\lambda)=\abs{A-\lambda I_N}$.

\begin{defn}
A matrix $M\in\mathscr M_{N\times N}$ is called \emph{non-negative} if $M_{ij}\geqslant 0$. A matrix $M\in\mathscr M_{N\times N}$ is called \emph{Metzler} if $M_{ij}\geqslant 0$ for all $j\neq i$. A matrix is said to be \emph{simple} if all of its diagonal entries are equal to zero.
\end{defn}

A \emph{weighted digraph} $\mathscr{G}=(V,A)$ is a 2-tuple consisting of a set of vertices or nodes $V=\{ 1,\ldots, N\}$ and an adjacency matrix $A\in\mathscr{M}_{N\times N}$ with the convention that there exists a directed edge from vertex $j$ to vertex $i$ if and only if $A_{ij}\neq0$, in which case $A_{ij}$ is the weight of the edge. We will always assume that $A_{ii}=0$, \emph{i.e.}, there are no self-loops in the digraph, so that every adjacency matrix considered is \emph{simple} as defined before. Given a node $i\in V$ of a weighted digraph $\mathscr{G}$, its \emph{weighted in-degree} is denoted by $\partial_i^-:=\sum_{j}A_{ij}$. The \emph{in-degree matrix} $D^-$ of a weighted digraph $\mathscr{G}$ is a diagonal matrix defined by $D_{ii}^-=\partial^-_i$. The \emph{in-degree Laplacian matrix} of a weighted digraph $\mathscr{G}=(V,A)$, denoted as $L^-$, is defined by  $L^-:=D^--A$. {Throughout this paper we don't require the graphs to be undirected, that is, we don't assume that $A$ is symmetric. Therefore, the eigenvalues $\mu_1,\ldots,\mu_N$ of the Laplacian matrix $L^-$ may be complex and they all satisfy ${\rm Re}(\mu_i)\geq \mu_1=0$.}

\begin{defn}
A weighted digraph $\mathscr{G}=(V,A)$ is \emph{in-regular} if every node $i\in V$ has the same in-degree $d^-$. Under such a condition, $d^-$ will denote the \emph{global in-degree} of the weighted digraph. $\mathscr{G}$ is \emph{strongly connected} if, for any two nodes $i$ and $j$, there exists a directed path which connects $i$ to $j$; in this case, its adjacency matrix is said to be \emph{irreducible}.
\end{defn}

\section{A network of slow-fast oscillators}

We present a single, general model that includes diffusive and excitatory synaptic coupling between slow-fast damped or sustained oscillators
\begin{equation}\label{eq:gral}
\begin{split}
\dot x_i=&-x_i-y_i+\!\sum_{j=1}^N A^d_{ij}(x_j\!-\!x_i)\!+\!S\left(\!\!\alpha_i x_i\!+\!\sum_{j=1}^N A^e_{ij}x_j\!\right),\\  
\dot y_i=&\varepsilon(x_i-y_i),    
\end{split}
\end{equation}
for every $i\in V=\{1,\ldots,N\}$, where $\varepsilon\in(0,1)$ is the time constant of the slow variables $y_i$, $\alpha_i>0$ are cellular positive feedback gains, and $S$ is a sigmoid function modeling intrinsic and synaptic nonlinearities. $A^d$ is the diffusive coupĺing adjacency matrix, and $A^e$ is the excitatory coupling adjacency matrix. Clearly, $(\vo{x}_0,\vo{y}_0)=(\vo{0}_N,\vo{0}_N)$ constitutes an equilibrium point for system (\ref{eq:gral}). When $A^e=O_N$ the coupling between the oscillators is \emph{purely diffusive}, and when $A^d=O_N$ the coupling is \emph{purely excitatory}. The two matrices $A^d$ and $A^e$ define the diffusive $\mathcal G^d(V,A^d)$ and excitatory $\mathcal G^e(V,A^e)$ digraphs, respectively.

\subsection{Transitions from damped to sustained oscillations ruled by a Hopf bifurcation}\label{subsec:two}

We now show that the parameter $\alpha$ rules the transition from damped to sustained slow-fast oscillations for uncoupled oscillators. Consider model (\ref{eq:gral}) for $N=1$, which reduces to the single-oscillator model
\begin{align*}
  \dot x&=-x-y+S(\alpha x),\\
  \dot y&=\varepsilon(x-y),
\end{align*}
The Jacobian matrix evaluated at equilibrium is readily computed as
$J(0,0)=\begin{pmatrix}
\alpha-1 & -1\\
\varepsilon & -\varepsilon
\end{pmatrix}$,
which leads to the pair of eigenvalues 
$\lambda_{1,2}=\tfrac{\alpha -(1+\varepsilon)}{2}\pm\tfrac{\sqrt{(\alpha +\varepsilon-1)^2-4\varepsilon}}{2}$.
For
\begin{equation}\label{eq:alpha hopf}
\alpha_H=1+\varepsilon>0    
\end{equation}
 both eigenvalues are purely imaginary. Moreover, continuity of the discriminant function $\Delta$ guarantees $\lambda_{1,2}\in\mathds{C}\backslash\mathds{R}$ for $\alpha$ sufficiently close to $\alpha_H$. Furthermore, observe that
$$\dfrac{\partial\mathrm{Re}(\lambda_{1,2})}{\partial\alpha}(\alpha_H)=\dfrac{1}{2}\not=0.$$
Invoking~\cite[Theorem 3.5.2]{guckenheimer83}, we can conclude the existence of a simple Hopf bifurcation for $\alpha=\alpha_H$, at which the model transitions from damped ($\alpha<\alpha_H$) to sustained ($\alpha>\alpha_H$) oscillations. The global validity of this result can be proved via Lyapunov and Poincaré-Bendixon arguments \cite[Theorem 1.8.1]{guckenheimer83}. We do not include it here due to space limitations.

\subsection{In-phase oscillations from network positive feedback between two coupled oscillators}

Consider model (\ref{eq:gral}) in the low-dimensional case $N=2$, $\alpha_1=\alpha_2=0$,  $A^d=O_2$, and
$$A^e=\left(
\begin{array}{cc}
    0 & \beta_2 \\
    \beta_1 & 0
\end{array}\right).$$
It is easy to show that the model undergoes a network Hopf bifurcation along the parametric curve $\sqrt{\beta_1\beta_2}=1+\varepsilon$, at which point the oscillators start to oscillate in phase, as shown in Fig. \ref{fig:exci1}. Observe that the uncoupled oscillators are damped in this case. The positive feedback brought by network interactions has the double role of both igniting and synchronizing the emergent oscillations.
\begin{figure}
    \centering
    \includegraphics[width=0.45\textwidth]{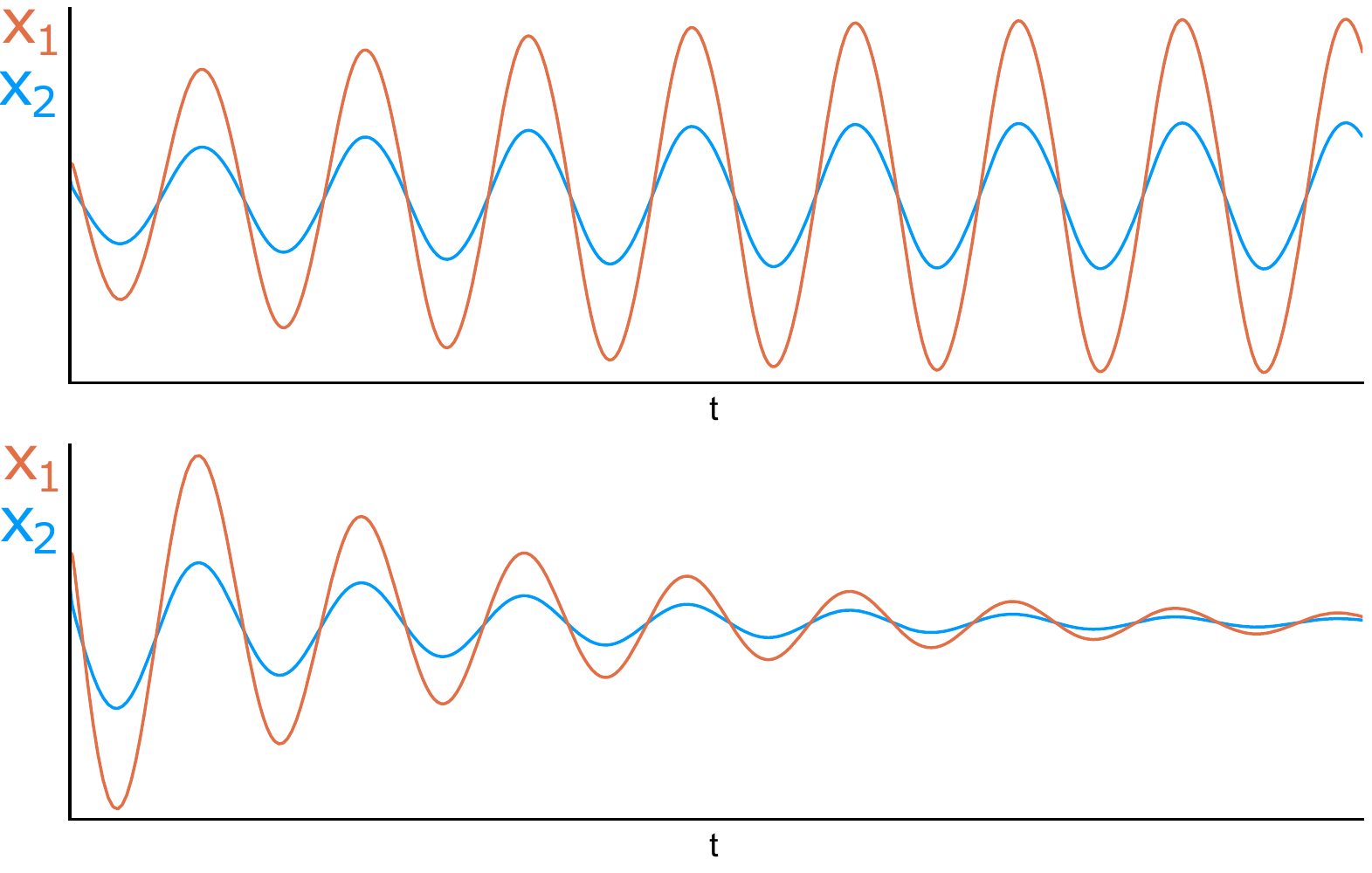}
    \caption{Model (\ref{eq:gral}) in the particular case $N=2$ and parameters as shown in Subsection \ref{subsec:two}. Variables $x_1$ and $x_2$ (in red and blue, respectively) are seen to oscillate in phase when excitatory parameters are taken near the curve $\beta_1\beta_2=(1+\varepsilon)^2$. Particular values for this graphs were $\varepsilon = 0.01$, $\beta_2 = 2.5$ and $\beta_1=\tfrac{(1+\varepsilon)^2}{\beta_2}+0.01$ (top) and $\beta_1=\tfrac{(1+\varepsilon)^2}{\beta_2}-0.01$ (bottom). Image generated using \texttt{Julia 1.5.2}.}
    \label{fig:exci1}
\end{figure}

We will show that the behavior observed in this low-dimensional example is impossible in general if the coupling is diffusive, whereas it generalizes to arbitrary strongly connected excitatory coupling topologies.

\section{A useful lemma}
During our discussion, we will find several block-wise defined matrices of the form
\begin{equation}\label{mat:J}
J=\left(\begin{array}{c|c}
    -aI_N+c M & -I_N \\
    \hline 
    \varepsilon I_N+d M & -\varepsilon I_N 
\end{array}\right),    
\end{equation}
where $M\in\mathscr{M}_{N\times N}$ is any real matrix, $\varepsilon>0$ is a (small) real constant, and $a,c,d\in\R$. The following general lemma will turn out very useful in our analysis.

\begin{lem}\label{lem:block}
Let $J\in\mathscr{M}_{2N\times 2N}$ be of form (\ref{mat:J}). Then its characteristic polynomial $p(\lambda)$, for $\lambda\not=-\varepsilon$, is obtained as
\begin{equation}\label{eq:lem1}
(\varepsilon+\lambda)^N\abs{\left(a+\lambda+\tfrac{\varepsilon}{\varepsilon+\lambda}\right)I_N-\left(c-\tfrac{d}{\varepsilon+\lambda}\right)M},    
\end{equation}
Moreover, any eigenvector $(\vo{x},\vo{y})\in\mathds{R}^{2N}$ of $J$, corresponding to an eigenvalue $\lambda\in\mathds{C}\backslash\{-\varepsilon\}$, must satisfy
\begin{equation}\label{eq:lem2}
\begin{split}
\vo{y}&=\tfrac{1}{\varepsilon+\lambda}(\varepsilon I_N+dM)\vo{x},\\ (c-\tfrac{d}{\varepsilon+\lambda})M\vo{x}&=(a+\lambda+\tfrac{\varepsilon}{\varepsilon+\lambda})\vo{x}.    
\end{split}
\end{equation}
\end{lem}
\begin{proof}
Obtaining the characteristic polynomial only requires us to apply the determinant formula for block-wise defined matrices \cite{powell11} to
$$\abs{J-\lambda I_{2N}}=\abs{\begin{array}{c|c}
    -(a+\lambda)I_N+c M & -I_N \\
    \hline 
    \varepsilon I_N+d M & -(\varepsilon+\lambda) I_N\end{array}}.$$
This will require $-(\varepsilon+\lambda)I_N$ to be invertible, which imposes $\lambda\not=-\varepsilon$. Thus, the formula calculates $p(\lambda)$ as
$$\abs{-(\varepsilon+\lambda)I_N}\abs{-(a+\lambda)I_N+cM-\tfrac{1}{\varepsilon+\lambda}(\varepsilon I_N+dM)},$$
whence (\ref{eq:lem1}) follows. As for the eigenvector condition (\ref{eq:lem2}), consider $\lambda\in\mathds{C}\backslash\{-\varepsilon\}$ a root of $p(\lambda)$ as given above, and suppose $(\vo{x},\vo{y})\in\mathds{R}^{2N}$ satisfies
$$\left(\begin{array}{c|c}
    -(a+\lambda)I_N+c M & -I_N \\
    \hline 
    \varepsilon I_N+d M & -(\varepsilon+\lambda) I_N\end{array}\right)
    \left(\begin{array}{c}
         \vo{x}  \\
         \hline
         \vo{y}
    \end{array}\right)=\left(\begin{array}{c}
         \vo{0}_N  \\
         \hline
         \vo{0}_N
    \end{array}\right),$$
This yields the linear system
\begin{align*}
    (-(a+\lambda)I_N+c M)\vo{x}-\vo{y}&=\vo{0}_N,\\
    (\varepsilon I_N+d M)\vo{x}-(\varepsilon+\lambda)\vo{y}&=\vo{0}_N.
\end{align*}
Given $\varepsilon+\lambda\not=0$, we may solve for $\vo{y}$ in the second equation as $\vo{y}=\tfrac{1}{\varepsilon+\lambda}(\varepsilon I_N+dM)\vo{x}$. We then substitute this into our first equation, getting 
$$(-(a+\lambda)I_N+c M)\vo{x}-\tfrac{1}{\varepsilon+\lambda}(\varepsilon I_N+dM)\vo{x}=\vo{0}_N.$$
From here the second eigenvector condition follows, thus concluding this proof.
\end{proof} 

The relevance of this lemma lies in that the matrix $M$ fully characterizes the spectral properties of the higher-dimensional matrix $J$. More precisely, equation (\ref{eq:lem1}) establishes a one-to-two correspondence between the eigenvalues of $J$ and those of $M$. Equation (\ref{eq:lem2}) establishes a similar correspondence between eigenvectors of these two matrices.

\section{Diffusive coupling cannot trigger sustained synchronous oscillations in networks of damped oscillators}

In this section we show that global rhythms are not sustainable within networks of damped nodes that are coupled diffusively (as a matter of fact, under such conditions, sustained synchronous oscillations are possible if individual feedback is high enough, that is, if every node is an intrinsic oscillator).


\begin{theo}\label{thm:diffusive}
Consider model (1) with $A^e=O_N$ and $\alpha_i=\alpha<1$ for all $i\in V=\{1,\ldots,N\}$, and $A^d\in\mathscr{M}_{N\times N}$ an arbitrary non-negative weighted adjacency matrix. Then, for sufficiently small $\varepsilon>0$ the origin is locally exponentially stable.
\end{theo}
\begin{proof}
Observe that
$$\dfrac{\partial\dot{x}_i}{\partial x_i} = \alpha S'(0)-1-\sum_{j\not=i} A_{ij}^d=\alpha S'(0)-1-\partial_i^-.$$
Thus, the model Jacobian computed at equilibrium is given by
\begin{align*}
J^d=J(\vo{0},\vo{0})&=\left(\begin{array}{c|c}
(\alpha-1)I_N-D^-+A^d & -I_N\\
\hline
\varepsilon I_N & -\varepsilon I_N
\end{array}\right)\\
&=\left(\begin{array}{c|c}
(\alpha -1)I_N-L^- & -I_N\\
\hline
\varepsilon I_N & -\varepsilon I_N
\end{array}\right),    
\end{align*}
where $L^-$ is the in-Laplacian matrix associated to $\mathcal G^d(V,A^d)$,
which is exactly in the form of Lemma \ref{lem:block}, with matrices $J=J^d$, $M=L^-$, and parameters $a=1-\alpha$, $c=-1$, $d=0$. The associated characteristic polynomial reads
$$p(\lambda)=(\varepsilon+\lambda)^N\abs{L^--(\alpha -1-\lambda-\tfrac{\varepsilon}{\varepsilon+\lambda})I_N}.$$
Let $\mu_1,\ldots,\mu_N$ be the eigenvalues of $L^-$ and recall that ${\rm Re}(\mu_i)\geq \mu_1=0$ for all $i\in\{1,\ldots,N\}$. Then any eigenvalue $\lambda$ of $J^d$ satisfies
$\alpha-1-\lambda-\tfrac{\varepsilon}{\varepsilon+\lambda}=\mu_k,$
which is equivalent to
{
\begin{equation}\label{eq:lambdamu}
\lambda^2+(\mu_k+1+\varepsilon-\alpha )\lambda+\varepsilon(\mu_k+2-\alpha )=0.  
\end{equation}
}
Thus, each $L^-$-eigenvalue $\mu_k$ yields two $J$-eigenvalues
$\lambda_{2k-1}=\lambda^-_k$ and $\lambda_{2k}=\lambda^+_k$, where
{
\begin{equation}\label{eq:valgral}
\lambda^{\pm}_k=\dfrac{\alpha\!-\!1\!-\!\varepsilon\!-\!\mu_k\pm\sqrt{(\mu_k+1-\varepsilon-\alpha)^2-4\varepsilon}}{2}  
\end{equation}
}
{for $k\in\{1,\ldots,N\}$. Setting $\varepsilon=0$ in (\ref{eq:lambdamu}) yields $\lambda_k^-=0$ with multiplicity $m=N$, and $\lambda_{k}^+=\alpha-1-\mu_k$, which satisfies $\mathrm{Re}(\lambda_{k}^+)=\alpha-1-\mathrm{Re}(\mu_k)<0$. By continuity, the real parts of eigenvalues $\lambda_{k}^+$ remains negative for sufficiently small $\varepsilon>0$. To guarantee a similar result for $\lambda_{k}^-$,
we can split (\ref{eq:lambdamu}) into its real and imaginary parts. Letting $\lambda=\sigma+i\tau$ and $\mu=u+iv$, we get
\begin{align*}
    &\sigma^2-\tau^2+\sigma(u+\varepsilon+1-\alpha)-v\tau+\varepsilon(2-\alpha+u)=0,\\
    &2\sigma\tau+\sigma v+\tau(u+\varepsilon+1-\alpha)+\varepsilon v=0,
\end{align*}
which can be interpreted as zero-level sets of some functions $F$, $G$, respectively. Using the Implicit Function Theorem, variables $\sigma$ and $\tau$ can be expressed as functions $S$, $T$ of the remaining variables $\varepsilon$, $u$, $v$ whenever
$$\dfrac{\partial(F,G)}{\partial(\sigma,\tau)}=(2\sigma+(u+\varepsilon+1-\alpha))^2+(2\tau+v)^2\not=0$$
is satisfied. The derivative $\tfrac{\partial S}{\partial \varepsilon}$ is readily obtained by implicit differentiation as
$$\dfrac{(\alpha-\sigma-2-u)(2\sigma+1+\varepsilon+u-\alpha)-(2\tau+v)(\tau+v)}{(2\sigma+(u+\varepsilon+1-\alpha))^2+(2\tau+v)^2}$$
which is negative for $\varepsilon=\sigma=\tau=0$, given that $u=Re(\mu)\geqslant0$ and $\alpha<1$. recall that $\lambda^-_{k}=0$ is obtained as a zero of (\ref{eq:lambdamu}) when $\varepsilon=0$. Thus,
for positive and {sufficiently small} values of $\varepsilon$, every $J^d$-eigenvalue has negative real part, and the equilibrium at the origin is locally exponentially stable.}  
\end{proof}

Theorem \ref{thm:diffusive} shows that diffusive coupling requires intrinsic oscillators to lead to synchronous network oscillations. One could provide a global proof by means of Lyapunov functions and convergent systems analysis \cite{pavlov05}. We omit this proof due to space constraints.

\section{Network and cellular positive feedback cooperate in triggering synchronous oscillations in networks of slow-fast damped nodes}

We now turn to the network positive feedback present in model (\ref{eq:gral}), for $A^d=O_N$ and non-negative $A^e$. Throughout this section we will make the standing \emph{homogeneity} assumption $\alpha_i=\alpha$ for every $i\in\{1,\ldots,N\}$, \emph{i.e.}, we assume that the uncoupled oscillators are identical. We will relax this homogeneity assumption in future works. We also let $A^e=\beta A$, where $A$ is a simple matrix. The two parameters $\alpha>0$, $\beta>0$ govern cellular and network positive feedback, respectively. Then, the Jacobian of model (\ref{eq:gral}) evaluated at its equilibrium at the origin reads 
$$J^e=J(\vo{0}_N,\vo{0}_N)=\left(\begin{array}{c|c}
-(1-\alpha)I_N+\beta A     & -I_N \\
\hline
\varepsilon I_N     & -\varepsilon I_N
\end{array}\right). $$
So we may apply Lemma \ref{lem:block} to matrix $J=J^e$, considering $M=A$, $a=1-\alpha$, $c=\beta$, $d=0$, to arrive at the following result.

\begin{lem}\label{lem:feed}
Let $A\in\mathscr{M}_{N\times N}$ be a simple, non-negative matrix, and consider model (\ref{eq:gral}) with $A^d=O_N$, $A^e=\beta A$ and $\alpha_i=\alpha$, where $\alpha\geqslant0$ and $\beta>0$ are non-negative parameters. Let $J^e$ be the Jacobian matrix of this system evaluated at equilibrium $(\vo{0}_N,\vo{0}_N)$. Then any eigenvector $(\vo{x},\vo{y})$ of $J^e$, associated to eigenvalue $\lambda$, must satisfy the conditions
\begin{equation}\label{eq:vect}
    \vo{y}=\tfrac{\varepsilon}{\varepsilon+\lambda}\vo{x},\ \ \beta A\vo{x}=(1-\alpha+\lambda+\tfrac{\varepsilon}{\varepsilon+\lambda})\vo{x}.
\end{equation}
Moreover, by letting $\mu_1,\ldots,\mu_N$ be the eigenvalues of matrix $\beta A$, we obtain for each $\mu_k$ two $J^e$-eigenvalues $\lambda_k^{\pm}$, where
\begin{equation}\label{eq:val}
\lambda^{\pm}_k\!=\!\dfrac{\mu_k\!+\!\alpha\!-\!1\!-\!\varepsilon\!\pm\sqrt{(\mu_k\!+\!\alpha\!-\!1\!-\!\varepsilon)^2\!-\!4\varepsilon(2\!-\!\alpha\!-\!\mu_k)}}{2}.
\end{equation}
\end{lem}
\begin{proof}
Apply Lemma \ref{lem:block} with $J=J^e$, $M=\beta A$, $a=1-\alpha$, $c=1$, $d=0$. Then equation (\ref{eq:vect}) holds. This implies that the eigenvalues of $\beta A$ and $J^e$ are linked by the expression $1-\alpha+\lambda+\tfrac{\varepsilon}{\varepsilon+\lambda}=\mu_k,$
which is equivalent to $\lambda^2+(1+\varepsilon-\alpha-\mu_k)\lambda+(2-\alpha-\mu_k)\varepsilon=0$.
From here we obtain two $J^e$-eigenvalues $\lambda^{\pm}_{k}$, which are indeed given by expression (\ref{eq:val}), thus ending the proof. 
\end{proof}

The relationship established at the end of Lemma \ref{lem:block} has become clearer, in that equation (\ref{eq:val}) explicitly determines $J^e$-eigenvalues as functions of $\beta A$-eigenvalues (and, therefore, of $A$-eigenvalues).
Thus, spectral analysis of matrix $J^e$, \emph{i.e.} local analysis of purely excitatory system (\ref{eq:gral}) under homogeneity hypotheses, reduces to spectral analysis of adjacency matrix $A$.

\subsection{In-regular homogeneous network}

We start by showing that if the coupling topology is in-regular, then model (\ref{eq:gral}) undergoes a Hopf-bifurcation for strong enough cellular and network positive feedback. Furthermore, because $\vo{1}_N$ is the dominant eigenvector of the adjacency matrix, the Hopf bifurcation happens along the synchronization space where each oscillator has the same state. That is, the network Hopf bifurcation leads to synchronous network oscillations.


\begin{theo}\label{theo:inreg}
Let $A\in\mathscr{M}_{N\times N}$ be an irreducible, simple, non-negative matrix associated to a strongly connected in-regular digraph of global in-degree $d^->0$, and consider model (\ref{eq:gral}) with $A^d=O_N$, $A^e=\beta A$ and $\alpha_i=\alpha$, where $\alpha\in[0,1)$ and $\beta>0$. Then, for sufficiently small $\varepsilon>0$ the system  undergoes a Hopf bifurcation along the parametric curve $\beta=\tfrac{1+\varepsilon-\alpha}{d^-}$. Moreover, the center manifold associated to the bifurcation is tangent to the synchronization subspace
$$E = \{(r\vo{1}_N,(\varepsilon r+\sqrt{\varepsilon(1-\varepsilon)}s)\vo{1}_N)\in\mathds{R}^{2N}:\,(r,s)\in\mathds{R}^2 \}$$
and is locally exponentially stable.
\end{theo}

\begin{proof} Given $A^e=\beta A$, where $A$ is in-regular, we conclude that $\mu_1=\beta d^->0$ is an eigenvalue with corresponding eigenvector $\vo{x}_1=\vo{1}_N$. Applying formula (\ref{eq:val}) to this eigenvalue yields two $J^e$-eigenvalues, namely $\lambda^{\pm}_{1}$, given by the expression
$$ \dfrac{\beta d+\alpha-1-\varepsilon\pm\sqrt{(\beta d+\alpha-1-\varepsilon)^2-4\varepsilon(2-\alpha-\beta d)}}{2}.$$
Thus for $\beta=\tfrac{1+\varepsilon-\alpha}{d^-}>0$, $\lambda^{\pm}_{1}$ are purely imaginary complex conjugates while all other eigenvalues have negative real part. Indeed, irreducibility of matrix $A$ makes it possible to apply the Perron-Frobenius Theorem which guarantees that the dominant eigenvalue $d^->0$ has algebraic and geometric multiplicity one.
Transversality is also easily verified (we omit details here due to space constraint), which yields the Hopf bifurcation.
Conditions (\ref{eq:vect}) give us the eigenvectors $\vo{z}_{1,2}=(\vo{1}_N,(\varepsilon\mp i\sqrt{\varepsilon(1-\varepsilon)})\vo{1}_N)$ associated to eigenvalues $\lambda^{\pm}_{1}=\pm i\lambda$, $\lambda = \sqrt{\varepsilon(1-\varepsilon)}$. Now, identifying the real and imaginary parts $\vo{u}=(\vo{1}_N,\varepsilon\vo{1}_N)$, $\vo{v}=(\vo{0}_N,\sqrt{\varepsilon(1-\varepsilon)}\vo{1}_N)$ of the spanning vectors, it follows that
$J^e\vo{u}\mp iJ^e\vo{v}=J^e(\vo{u}\mp i\vo{v})=\pm i\lambda(\vo{u}\mp i\vo{v})=\lambda\vo{v}\pm i\lambda\vo{u}.$
$J^e$, $\vo{u}$, $\vo{v}$ and $\lambda$ are real (matrices, vectors and values), so this last equation implies $J^e(\vo{u})=\lambda\vo{v}$ and $J^e\vo{v}=-\lambda\vo{u}$. Therefore the center manifold is tangent to the span of vectors $\vo{u}$, $\vo{v}$, whence the form of subspace $E$ is obtained. Because all other eigenvalues have negative real part, the associated center is locally exponentially attractive. 
\end{proof}
\begin{figure}
    \centering
    \includegraphics[width=0.47\textwidth]{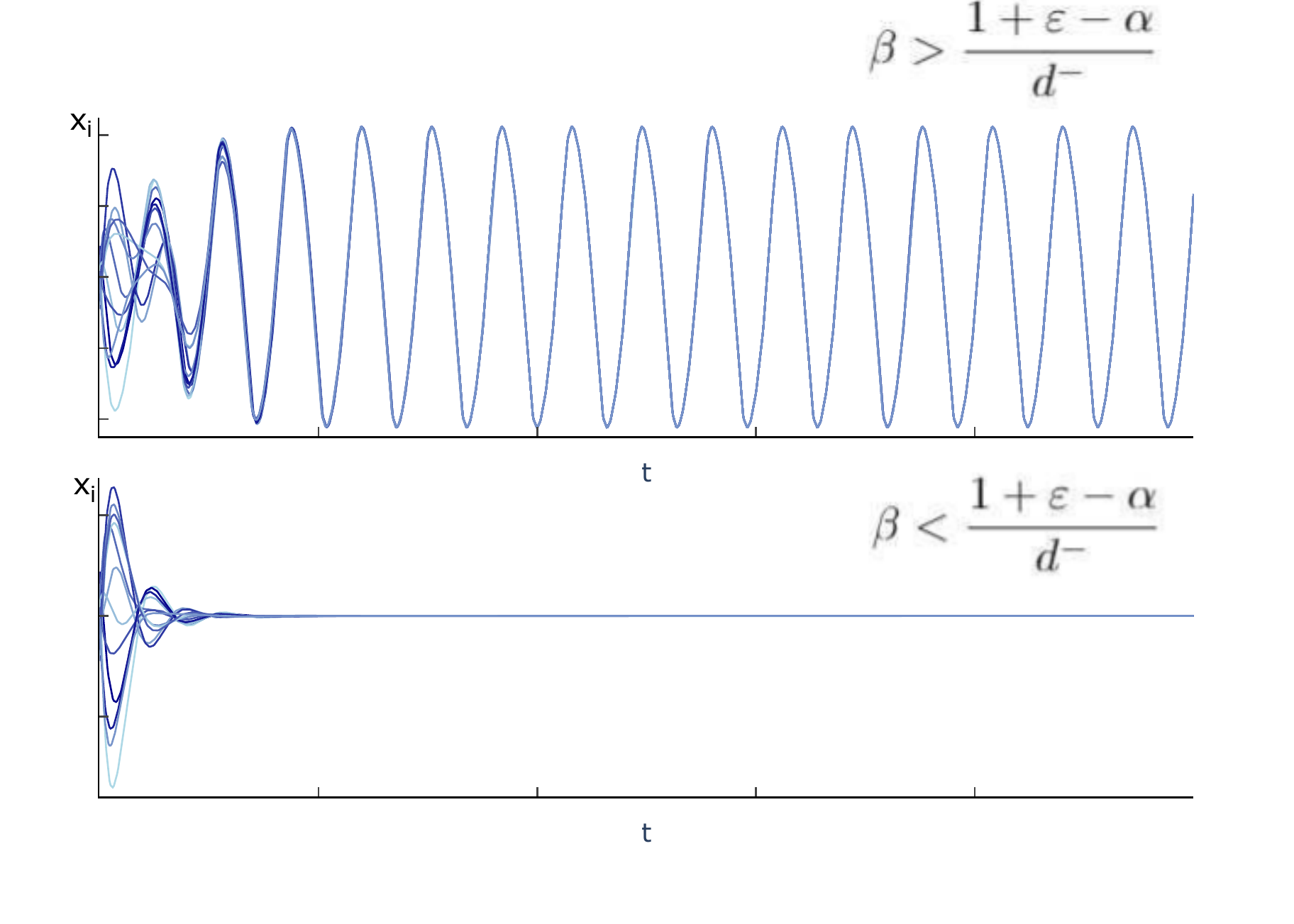}
    \caption{Contrast of two different configurations for model (\ref{eq:gral}) under in-regularity and homogeneity conditions as described in Theorem \ref{theo:inreg}. Matrix $A$ is a weighted modification of an adjacency matrix corresponding to the Frucht graph \cite{frucht39}, $N=12$. Specific parameters are $\alpha=0.5$, $\varepsilon=0.01$ and $d^-=5.0$. Image generated using \texttt{Julia 1.5.2}.}
    \label{fig:excifru1}
\end{figure}

Theorem \ref{theo:inreg} shows that when
$\beta-\frac{1+\varepsilon-\alpha}{d^-}>0$ the model exhibits synchronous oscillations. This condition can be fulfilled both by increasing the cellular positive feedback $\alpha$ for fixed network positive feedback $\beta$ or vice-versa. Figure \ref{fig:excifru1} numerically illustrates the predictions of our theorem.


\subsection{Strongly-connected homogeneous network}

In-regular networks are too restrictive to accurately model biological networks like the SCN. In this section we relax the in-regularity assumption. Irreducibility of the adjacency matrix for strongly connected coupling topologies implies uniqueness of a Perron-Frobenius eigenvector that fully determines the pattern of in-phase oscillations emerging at the network Hopf bifurcation.

\begin{theo}\label{theo:str}
Let $A\in\mathscr{M}_{N\times N}$ be an irreducible, simple, non-negative matrix associated to a strongly connected digraph, and consider model (\ref{eq:gral}) associated to $A^d=O_N$, $A^e=\beta A$ and $\alpha_i=\alpha$, $\alpha\in[0,1)$ and $\beta>0$. Let $\rho>0$ be the leading eigenvalue of $A$, and $\vo{x}_0>0$ the Perron eigenvector associated to $\rho$. Then, for $\varepsilon>0$ and sufficiently small, the system  undergoes a Hopf bifurcation along the parametric curve $\beta=\tfrac{1+\varepsilon-\alpha}{\rho}$. Moreover, the center manifold associated to the bifurcation is tangent to the real subspace 
$$E=\{(r\vo{x}_0,(r\varepsilon+s\sqrt{\varepsilon(1-\varepsilon)})\vo{x}_0)\in\mathds{R}^{2N}:\,(r,s)\in\mathds{R}^2\},$$
and is locally exponentially stable.
\end{theo}

\begin{proof}
Given $A^e = \beta A$, where $A$ is irreducible and non-negative, we conclude that $\mu_1=\beta\rho>0$ is its leading real eigenvalue with corresponding eigenvector $\vo{x}_0>0$. Applying formula (\ref{eq:val}) to this eigenvalue yields two $J^e$-eigenvalues $\lambda^{\pm}_{1}$, where
$$\lambda^{\pm}_1=\dfrac{\beta\rho\!+\!\alpha\!-\!1\!-\!\varepsilon\!\pm\!\sqrt{(\beta\rho\!+\!\alpha\!-\!1\!-\!\varepsilon)^2\!-\!4\varepsilon(2\!-\!\alpha\!-\!\beta\rho)}}{2}.$$
Thus, for $\beta=\tfrac{1+\varepsilon-\alpha}{\rho}>0$, $\lambda^{\pm}_{1}$ are purely imaginary complex eigenvalues while all other eigenvalues have negative real part. Indeed, irreducibility of matrix $A$ makes it possible to apply the Perron-Frobenius Theorem which guarantees that leading eigenvalue $\rho>0$ has algebraic and geometric multiplicity one. Transversality is once again easily verified, which yields the Hopf bifurcation. Setting $\lambda = \sqrt{\varepsilon(1-\varepsilon)}$, conditions (\ref{eq:vect}) give us the $J^e$-eigenvector $\vo{z}_{1,2}=(\vo{x}_0,(\varepsilon\mp i\sqrt{\varepsilon(1-\varepsilon)})\vo{x}_0)$ associated to $\lambda_{1}^{\pm}=\pm i\lambda$ at bifurcation. As in the previous Theorem, writing $\vo{z}_0=\vo{u}_0\mp i\vo{v}_0$ in its real and imaginary parts, one again sees that $\vo{u}_0$ and $\vo{v}_0$ span the real tangent subspace to the center manifold, whence we conclude the form of subspace $E$. Because all other eigenvalues have negative real part at the bifurcation, the associated center manifold is locally exponentially stable. 
\end{proof}
\begin{figure}
    \centering
    \includegraphics[width=0.47\textwidth]{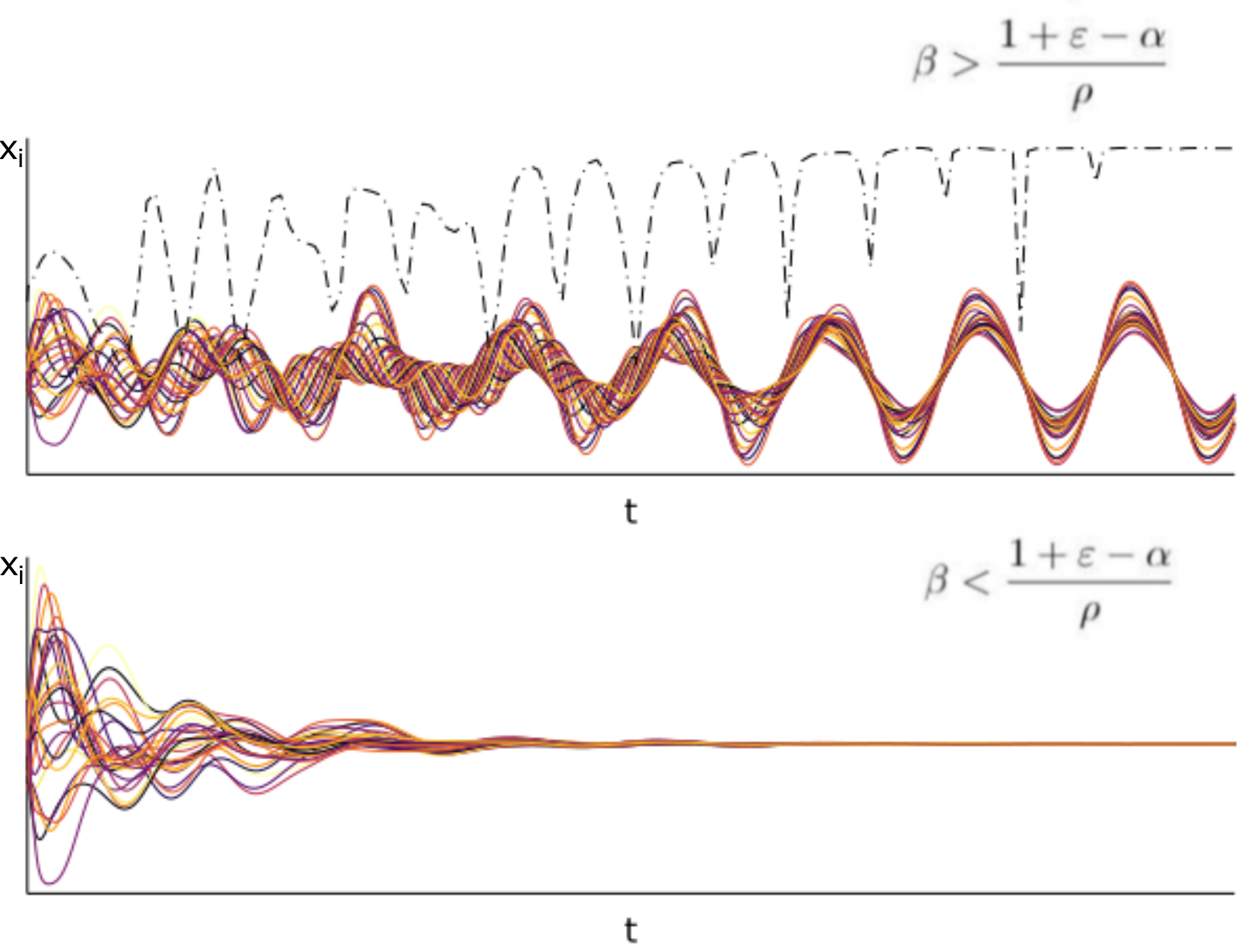}
    \caption{Contrast of two different configurations for model (\ref{eq:gral}) under homogeneity conditions as described in Theorem \ref{theo:str}. Matrix $A$ corresponds to a weighted directed cycle, $N=25$, with random positive weights $(d_1,\ldots,d_N)$. The leading eigenvalue $\rho>0$ is the positive solution of $r^N=\prod d_i$. Specific parameters are $\alpha=0.5$ and $\varepsilon=0.01$. {In the upper plot, the dashed black line shows the evolution of $l(t)=\frac{\abs{\langle\vo{x}_0,\vo{x}(t)\rangle}}{\|\vo{x}_0\|\|\vo{x}(t)\|}$, where $\langle\cdot,\cdot\rangle$ and $\|\cdot\|$ denote the standard scalar product and 2-norm, respectively, and $\vo{x}_0$ is the adjacency matrix Perron-Frobenious eigenvector as defined in Theorem~\ref{theo:str}. Observe that $l(t)\to 1$ at almost all time points, \emph{i.e.}, excluding time points where $\vo{x}(t)=0$}. In other words, along the in-phase network oscillations, the $\vo{x}$-component of the state vector is parallel to the Perron eigenvector $\vo{x}_0$, as predicted by Theorem \ref{theo:str}.}
    \label{fig:excistr1}
\end{figure}

Theorem \ref{theo:str} shows that when $\beta = \tfrac{1+\varepsilon-\alpha}{\rho}>0$ the model exhibits in-phase oscillations. This condition can be fulfilled both by increasing the cellular positive feedback $\alpha$ for fixed network positive feedback $\beta$ or vice-versa. Furthermore, it shows that the Perron-Frobenius eigenvector of the adjacency matrix fully controls the oscillation pattern, at least close to the Hopf bifurcation. Figure \ref{fig:excistr1} numerically illustrates the predictions of our theorem.

\section{Discussion and future directions}

\subsection{Model extension}

As it is the case for many physiological networks, sometimes outputs from one node affect the receiving node along multiple timescales. To incorporate such an effect in our model one could consider an \emph{extended case} of the excitatory version, adjusting equations (\ref{eq:gral}) to account for the effect of $\vo{x}$-variables over $\vo{y}$-variables. Under homogeneity conditions, computations of this extended case would be very similar to those made before, as seen by applying Lemma \ref{lem:block} when $d\not=0.$

\subsection{Extension to global results}

The analysis made in this paper relies solely on local properties of dynamical systems near equilibrium. Therefore, a more complete and formal approach should also incorporate global tools before and after bifurcation to guarantee convergence to either a stable steady state or a stable limit cycle. For example, in the diffusive case one could propose a Lyapunov function \cite{zhou11} or use the theory of convergent systems \cite{pavlov05}. Alternatively one can use dominance analysis~\cite{forni2018differential}, through which it might be possible to show the existence of a globally attractive and invariant 2-dimensional manifold corresponding to the center manifold of the Hopf bifurcation, which would effectively make our local bifurcation analysis global.

\subsection{Heterogeneous populations}

It is evident that real life networks won't maintain, in general, the homogeneous hypothesis which were used extensively in the proofs of the excitatory case. A more delicate analysis should be provided when considering heterogeneous networks, much more likely to be found in real life phenomena, through higher dimensional bifurcation theory \cite{guckenheimer83}.

\subsection{Application to circadian rhythmogenesis}

The model in this work was originally motivated by the synchronization phenomena observed in the suprachiasmatic nucleus (SCN) of the mammal hypothalamus. One may identify different subpopulations and connections inside the SCN (e.g. spatially \cite{welsh10}, GABAergic \cite{dewoskin15}, neuropeptidergic \cite{evans16, shan20}). Neuromodulation here not only affects the electrophysiological rhythms, but also gives input to the molecular clock through a much slower loop \cite{diekman13}. Other works have pointed out that neural appositions and density of connections vary according from one neuropeptidergic subpopulation to the other \cite{mieda19, varadarajan18}. Therefore, a \emph{multilayer digraph} model may result useful in capturing the dynamic properties of this circadian phenomenon. Additional topologies $A_j^e$ could be incorporated to account for different neuropeptide release (VIP, AVP, GRP being the main ones).

\bibliographystyle{ieeetr}
\bibliography{synchro_ref.bib}

\begin{thebibliography}{10}

\bibitem{kim16}
J.~Kim, J.~Yang, H.~Shim, J.~S. Kim, and J.~H. Seo, ``Robustness of
  synchronization of heterogeneous agents by strong coupling and a large number
  of agents,'' {\em IEEE Trans Automat Contr}, vol.~61, no.~10, pp.~3096--3102,
  2016.

\bibitem{panteley17}
E.~Panteley and A.~Loria, ``Synchronization and dynamic consensus of
  heterogeneous networked systems,'' {\em IEEE Trans Automat Contr}, vol.~62,
  no.~8, pp.~3758--3773, 2017.

\bibitem{arechiga04}
H.~Aréchiga, ``Sustrato neural de los ritmos biológicos,'' {\em Mensaje
  Bioquímico}, vol.~28, pp.~25--250, 2004.

\bibitem{evans16}
J.~A. Evans, ``Collective timekeeping among cells of the master circadian
  clock,'' {\em J. of Endocrinol}, vol.~230, pp.~27--49, 2016.

\bibitem{webb09}
A.~B. Webb, N.~Angelo, J.~E. Huettner, and E.~D. Herzog, ``Intrinsic,
  nondeterministic circadian rhythm generation in identified mammalian
  neurons,'' {\em PNAS}, vol.~106, no.~38, pp.~16493--16498, 2009.

\bibitem{noguchi17}
T.~Noguchi, T.~L. Leise, N.~J. Kingsbury, T.~Diemer, L.~L. Wang, M.~A. Henson,
  and D.~K. Welsh, ``Calcium circadian rhythmicity in suprachiasmatic nucleus:
  Cell autonomy and network modulation,'' {\em eNeuro}, vol.~4, 7 2017.

\bibitem{hh}
A.~L. Hodgkin and A.~F. Huxley, ``A quantitative description of membrane
  current and its application to conduction and excitation in nerve,'' {\em J
  Physiol}, vol.~117, no.~4, pp.~500--544, 1952.

\bibitem{fhn}
R.~FitzHugh, ``Mathematical models of threshold phenomena in the nerve
  membrane,'' {\em Bull Math Biophysics}, vol.~17, pp.~257--278, 2020.

\bibitem{gusev16}
V.~V. Gusev and E.~V. Pribavkina, ``On synchronizing colorings and the
  eigenvectors of digraphs,'' in {\em 41st International Symposium on
  Mathematical Foundations of Computer Science}, MFCS, 2016.

\bibitem{zhong17}
J.~Zhong, J.~Lu, and D.~W.~C. Ho, ``Controllability and synchronization
  analysis of identical-hierarchy mixed-value logical control networks,'' {\em
  IEEE Trans Cybern}, vol.~47, no.~11, pp.~3482--3493, 2017.

\bibitem{lee20b}
J.~G. Lee and R.~Sepulchre, ``Rapid synchronization under weak synaptic
  coupling,'' in {\em 59th IEEE Conference on Decision and Control}, CDC, 2020.

\bibitem{forni2018differential}
F.~Forni and R.~Sepulchre, ``Differential dissipativity theory for dominance
  analysis,'' {\em IEEE Transactions on Automatic Control}, vol.~64, no.~6,
  pp.~2340--2351, 2018.

\bibitem{guckenheimer83}
J.~Guckenheimer and P.~Holmes, {\em Nonlinear Oscillations, Dynamical Systems,
  and Bifurcations of Vector Fields}.
\newblock Springer, 1983.

\bibitem{powell11}
P.~D. Powell, ``Calculating determinants of block matrices.''
  https://arxiv.org/pdf/1112.4379v1.pdf, 2011.

\bibitem{pavlov05}
A.~Pavlov, N.~de~Wouw, and H.~Nijmeijer, ``Convergent systems: Analysis and
  synthesis,'' in {\em Control and Observer Design for Nonlinear Finite and
  Infinite Dimensional Systems} (T.~Meurer, K.~Graichen, and E.-D. Gilles,
  eds.), pp.~131--146, Springer, 2005.

\bibitem{frucht39}
R.~Frucht, ``Herstellung von graphen mit vorgegebener abstrakter gruppe,'' {\em
  Compositio Mathematica}, vol.~6, pp.~239--250, 1939.

\bibitem{zhou11}
X.~Zhou, H.~Feng, J.~Feng, and Y.~Zhao, ``On synchonization of
  pinning-controlled networks with reducible and asymmetric coupling matrix,''
  {\em Communications and Networks}, vol.~3, no.~2, pp.~118--126, 2011.

\bibitem{welsh10}
D.~K. Welsh, J.~S. Takahashi, and S.~A. Kay, ``Suprachiasmatic nucleus: Cell
  autonomy and network properties,'' {\em Annu Rev Physiol.}, vol.~72,
  pp.~551--577, 2010.

\bibitem{dewoskin15}
D.~DeWoskin, J.~Myung, M.~C.~D. Belle, H.~D. Piggins, T.~Takumi, and D.~B.
  Forger, ``Distinct roles for gaba across multiple timescales in mammalian
  circadian rhythm,'' {\em PNAS}, 2015.

\bibitem{shan20}
Y.~Shan, J.~H. Abel, Y.~Li, D.~P. Olson, and J.~S. Doyle, F. J. III
  ans~Takahashi, ``Dual-color single-cell imaging of the suprachiasmatic
  nucleus reveals a circadian role in network synchrony,'' {\em Neuron},
  vol.~108, pp.~1--16, 10 2020.

\bibitem{diekman13}
C.~O. Diekman, M.~D.~C. Belle, R.~P. Irwin, C.~N. Allen, and H.~D. Piggins,
  ``Causes and consequences of hyperexcitation in central clock neurons,'' {\em
  PLoS Comput Biol}, vol.~9, 2013.

\bibitem{mieda19}
M.~Mieda, ``The network mechanism of the central circadian pacemaker of the
  scn: Do avp neurons play a more critical role than expected?,'' {\em Front
  Neurosci.}, vol.~13, no.~139, 2019.

\bibitem{varadarajan18}
S.~Varadarajan, R.~Tajiri, M. ans~Jain, R.~Holt, Q.~Ahmed, J.~LeSauter, and
  R.~Silver, ``Connectome of the suprachiasmatic nucleus: New evidence of the
  core-shell relationship,'' {\em eNeuro}, vol.~5, no.~5, 2018.

\end{thebibliography}

\addtolength{\textheight}{-12cm}   

\end{document}